\newif\ifSOSA
\SOSAfalse

\ifSOSA
\documentclass[twoside,leqno]{article}
\usepackage[letterpaper]{geometry}
\usepackage{ltexpprt}
\else
\documentclass[11pt]{article}
\usepackage[utf8]{inputenc}
\usepackage[margin=1in]{geometry}
\usepackage{amsmath,amsfonts,amssymb}
\usepackage{amsthm}

\newtheorem{theorem}{Theorem}[section]
\newtheorem{corollary}[theorem]{Corollary}
\newtheorem{lemma}[theorem]{Lemma}

\newtheorem{proposition}[theorem]{Proposition}
\newtheorem{claim}[theorem]{Claim}
\theoremstyle{definition}
\newtheorem{definition}[theorem]{Definition}

\fi

\usepackage{hyperref}

\usepackage{amssymb}
\usepackage[noadjust]{cite}
\usepackage{graphicx}
\usepackage{hyperref}
\usepackage{mathtools}
\usepackage{xcolor}
\usepackage{dirtytalk}
\usepackage{thm-restate}

\newcommand{\N}{\mathbb{N}}

\newcommand{\A}{\mathcal{A}}
\newcommand{\R}{\mathbb{R}}

\newcommand{\ep}{\eps}
\newcommand{\ra}{\rightarrow}
\newcommand{\la}{\leftarrow}

\newcommand{\Va}{V_{\mathit{acc}}}
\newcommand{\va}{v_{\mathit{acc}}}
\newcommand{\vs}{v_{0}}

\newcommand{\fS}{\mathcal{S}}
\newcommand{\poly}{\operatorname{poly}}

\renewcommand{\L}{\mathbf{L}}

\DeclareMathOperator{\samp}{SAMP}

\newcommand{\eps}{\varepsilon}
\newcommand{\tO}{\tilde{O}}

\newcommand{\zo}{\{0,1\}}

\newcommand{\Vn}{V_{\text{neg}}}
\newcommand{\F}{\mathcal{F}}

\newif\ifdraft
\draftfalse

\title{\Large Deterministic Approximation of Random Walks via Queries in Graphs of Unbounded Size}
\author{Edward Pyne\thanks{Supported by NSF grant CCF-1763299.} \\ Harvard University \\ \texttt{epyne@college.harvard.edu} \\ \and Salil Vadhan\thanks{Supported by NSF grant CCF-1763299 and a Simons Investigator Award.} \\ Harvard University \\ \texttt{salil\char`_vadhan@harvard.edu}}
\begin{document}
\begin{titlepage}
\ifSOSA
\fancyfoot[R]{\scriptsize{Copyright \textcopyright\ 2022 by SIAM\\
Unauthorized reproduction of this article is prohibited}}
\fi
\maketitle
\begin{abstract}
    Consider the following computational problem: given a regular digraph $G=(V,E)$, two vertices $u,v \in V$, and a walk length $t\in \mathbb{N}$, estimate the probability that a random walk of length $t$ from $u$ ends at $v$ to within $\pm \varepsilon.$  A randomized algorithm can solve this problem by carrying out $O(1/\varepsilon^2)$ random walks of length $t$ from $u$ and outputting the fraction that end at $v$.

    In this paper, we study $\textit{deterministic}$ algorithms for this problem that are also restricted to carrying out walks of length $t$ from $u$ and seeing which ones end at $v$. Specifically, if $G$ is $d$-regular, the algorithm is given oracle access to a function $f : [d]^t\to \{0,1\}$ where $f(x)$ is $1$ if the walk from $u$ specified by the edge labels in $x$ ends at $v$.  We assume that G is $\textit{consistently labelled}$, meaning that the edges of label $i$ for each $i\in [d]$ form a permutation on $V$. 

    We show that there exists a deterministic algorithm that makes $\text{poly}(dt/\varepsilon)$ nonadaptive queries to $f$, regardless of the number of vertices in the graph $G$.  Crucially, and in contrast to the randomized algorithm, our algorithm does not simply output the average value of its queries.  Indeed, Hoza, Pyne, and Vadhan (ITCS 2021) showed that any deterministic algorithm of the latter form that works for graphs of unbounded size must have query complexity at least $\exp(\tilde{\Omega}(\log(t)\log(1/\varepsilon)))$.  

    In the language of pseudorandomness, our result is a separation between the query complexity of ``deterministic samplers'' and ``deterministic averaging samplers'' for the class of ``permutation branching programs of unbounded width''.  Our separation is stronger than the prior separation of Pyne and Vadhan (CCC 2021), and has a much simpler proof (not using spectral graph theory or the Impagliazzo--Nisan--Wigderson pseudorandom generator).  On the other hand, the algorithm of Pyne and Vadhan is explicit and computable in small space, whereas ours is not explicit (unless we assume the existence of an optimal explicit pseudorandom generator for permutation branching programs of bounded width).
\end{abstract}

\vfill
\thispagestyle{empty}
\end{titlepage}
\newpage
\section{Introduction}
    Consider the following computational problem: given a regular digraph $G=(V,E)$, two vertices $u,v \in V$, and a walk length $t\in \N$, estimate the probability that a random walk of length $t$ from $u$ ends at $v$ to within $\pm \ep.$  A randomized algorithm can solve this problem by carrying out $O(1/\ep^2)$ random walks of length $t$ from $u$ and outputting the fraction that end at $v$. 

    In this paper, we study {\em deterministic} algorithms for this problem that are also restricted to carrying out walks of length $t$ from $u$ and seeing which ones end at $v$.  Specifically, if $G$ is $d$-regular, the algorithm is given oracle access to a function $f : [d]^t\ra\zo$ where $f(x)$ is 1 if the walk from $u$ specified by the edge labels in $x$ ends at $v$.  We assume that G is {\em consistently labelled}, meaning that the edges of label $i$ for each $i\in [d]$ form a permutation on $V$.  (It can be shown that every 
    $d$-regular digraph has a consistent labelling, so this requirement does not constrain the graph structure, only the labelling.) In the case where we have ``white-box'' access to the graph, Ahmadinejad, Kelner, Murtagh, Peebles, Sidford, and Vadhan~\cite{AhmadinejadKeMuPeSiVa19} gave a deterministic algorithm for this problem with space complexity $\tO(\log(|V|\cdot td/\ep))$.

    We show that there exists a deterministic algorithm that makes $\poly(dt/\ep)$ nonadaptive queries to $f$, regardless of the number of vertices in the graph $G$.  Crucially, and in contrast to the randomized algorithm, our algorithm does not simply output the average value of its queries.  Indeed, Hoza, Pyne, and Vadhan~\cite{HozaPyVa21} showed that any deterministic algorithm of the latter form that works for graphs of unbounded size must have query complexity at least $\exp(\tilde{\Omega}(\log (t)\log(1/\ep)))$.   

    Below we present these results in the language of pseudorandomness, as a separation between the query complexity of ``deterministic samplers'' and ``deterministic averaging samplers'' for the class of ``permutation branching programs of unbounded width''.  Our separation is stronger than the prior separation of Pyne and Vadhan~\cite{PyneVa21a}, and has a much simpler proof (not using spectral graph theory or the Impagliazzo--Nisan--Wigderson pseudorandom generator).  On the other hand, the algorithm of Pyne and Vadhan is explicit and computable in small space, whereas ours is not explicit (unless we assume the existence of an optimal explicit pseudorandom generator for permutation branching programs of bounded width).
    
\subsection{Ordered Branching Programs}
    Motivated by the goal of derandomizing space-bounded computation, i.e. proving $\textbf{BPL}=\L$, there has been extensive work on estimating the acceptance probabilities of \textit{ordered branching programs}, which capture how a randomized small-space algorithm uses its random bits.
    \begin{definition}
        An \textbf{ordered branching program (OBP)} $B$ of \textbf{length} $n$ and \textbf{width} $w$ computes a function $B:\zo^n\ra\zo$. On an input $\sigma\in \zo^n$, the branching program computes as follows. It has $n+1$ layers $V_0,\ldots,V_n$, each with vertices labeled $\{1,\ldots,w\}$. It starts at a fixed start state $\vs\in V_0$. 
        Then for $r=1,\ldots,n$, it reads the next symbol $\sigma_r$ and updates its state according to a transition function $B_r:V_{r-1}\times \zo\ra V_r$ by taking $v_{r}=B_r(v_{r-1},\sigma_r)$. 
        For $v\in V_i$ and $u\in V_j$ for $j>i$, we write $B[v,x]=u$ if the program transitions to state $u$ starting from state $v$ on input $x\in \zo^{j-i}$.

	    Moreover, there is an accept state $\va\in V_n$.
	    For $x\in \zo^n$, we define $B(x)=1$ iff $B[\vs,x]=\va$.  That is, $B$ accepts the inputs $x$ that lead it from the start state $\vs$ in the first layer to the accept state in the last layer $\va$.
    \end{definition}
    
    Given a graph $G=(V,E)$ with $w$ vertices, outdegree 2, $n\in \N$, and two vertices $u,v\in V$, we can define an associated ordered branching program $B^{G,u,v,n}$ which simulates walks of length $n$ in $G$.  Specifically, we set $\vs=u$, $\va=v$, and $B^{G,u,v,n}_r(v_{r-1},\sigma)$ to be the $\sigma$'th neighbor of vertex $v_{r-1}$ for every $r=1,\ldots,n$, $v_{r-1}\in V_{r-1}=\{1,\ldots,w\}$, and $\sigma\in \{0,1\}$.  Then $B^{G,u,v,n}(x)=1$ iff carrying out a walk of length $n$ according to the edge labels in $x$ leads from $u$ to $v$ in $G$.  In particular, the probability that $B^{G,u,v,n}$ accepts a uniformly random input $x$ is exactly the probability that a random walk of length $n$ from $u$ ends at $v$, exactly the computational problem we wish to solve.
    Compared to a general ordered branching program, $B^{G,u,v,n}$ has the same transition function at every layer.
    
    The definition of branching programs naturally generalizes to alphabet sizes $d>2$, allowing for simulation of random walks on graphs of degree $d$.  We present our results just for the $d=2$ case for simplicity, but they extend to larger $d$ as well.

    Motivated by the derandomization of space-bounded computation, there has been three decades of work on deterministically estimating the acceptance probability of ordered branching programs in small space (for instance~\cite{BabaiNiSz92,Nisan92,ImpagliazzoNiWi94,SaksZh99,BravermanCoGa18,MekaReTa18,Hoza21} and many others). In the case where we have white-box access to the branching program, the algorithm of Saks and Zhou~\cite{SaksZh99}, as recently improved by Hoza~\cite{Hoza21}, achieves space complexity $o(\log^{3/2}n)$ in the case $w=\poly(n)$ and $\ep=1/\poly(n)$.
    
    Our focus, however, is on ``black-box'' derandomization, where we only have oracle access to the function $B : \zo^n\rightarrow \zo$. 
    In this setting, we consider two questions: 
    \begin{enumerate}
        \item Is there a small set $Q\subseteq \zo^n$ such that knowing the value of an arbitrary branching program $B$ on all points of $Q$ allows us to estimate $\Pr[B(U_n)=1]$ up to additive error $\ep$? 
        We call the size of the smallest such $Q$ the \textbf{query complexity} of two-sided derandomization.
        \item If so, can we explicitly construct this set, and compute the estimate of the probability, in space $O(\log |Q|)$?
    \end{enumerate}
    
    Such algorithms can be seen as deterministic analogues of {\em samplers} for a restricted class of functions.  (See the 
    survey \cite{Goldreich11} for a general treatment of samplers.)  These are defined as follows, following the language of Cheng and Hoza~\cite{ChengHo20}:
    \begin{definition}
        Let $\F$ be a class of functions $f : \zo^n \ra \zo$. A \textbf{deterministic $\ep$-sampler} for $\mathcal{F}$ is an algorithm $\samp$ that, given oracle access to an arbitrary $f\in \F$, makes queries to $f$ and outputs
        $\samp^f()\in \R$ such that
        \[|\samp^f()-\Pr[f(U_n)=1]|\leq \ep.\]
        We say the the \textbf{query complexity} $S$ is the maximum over $f\in \F$ of the number of distinct queries made. We say the sampler is \textbf{explicit} if, given $n$, $\eps$, and parameters defining the family $\F$ and $\eps>0$, $\samp$ can be computed by a uniform algorithm with space complexity $O(\log S)$.
        
        We say that $\samp$ is \textbf{nonadaptive} if it makes nonadaptive queries to its oracle $f$.  A special case of nonadaptive samplers are \textbf{deterministic averaging samplers} whereby the output of the sampler is the average of $f$ over its set $Q$ of oracle queries.
    \end{definition}
    It can be shown that deterministic averaging $\ep$-samplers for a class $\F$ are equivalent to $\ep$-pseudorandom generators (PRGs) for the class $\F$, where the seed length of the PRG is equal to the logarithm of the query complexity of the sampler.
    
    We also consider deterministic hitters (i.e. hitting set generators), a weaker one-sided analogue of deterministic samplers.
    \begin{definition}\label{def:hsg}
	    Let $\F$ be a class of functions $f : \zo^n \ra \zo$. A \textbf{deterministic $\ep$-hitter} for $\F$ is a set $H\subseteq \zo^n$ such that for every $f \in \F$ where $\Pr_{x \la U_{n}}[f(x)=1]> \ep$, there is $x \in H$ such that $f(x)=1$. We say that $H$ is \textbf{explicit} if the elements of $H$ can be enumerated in space $O(\log |H|)$.
    \end{definition}
    It can be shown that a deterministic $\ep$-sampler $\samp$ for a class $\F$ implies a deterministic $2\ep$-hitter $H$ for $\F$, where $H$ is explicit if $\samp$ is. (Let $H$ be the queries made by $\samp^{f_0}$ on the all-zeroes function $f_0$.)
    
    A standard application of the probabilistic method (see Appendix~\ref{app:proofs}) shows that there do exist deterministic averaging samplers (i.e. PRGs) with polynomial query complexity for ordered branching programs of polynomial size.  Specifically, for ordered branching programs of length $n$ and width $w$, there exists a deterministic averaging $\ep$-sampler with query complexity $\poly(nw/\ep)$, and this is optimal.
    However, constructing an \textit{explicit} deterministic sampler with matching query complexity has been a longstanding open problem. The classic deterministic averaging samplers (i.e. PRGs) of Nisan~\cite{Nisan92} and Impagliazzo--Nisan--Wigderson~\cite{ImpagliazzoNiWi94} have query complexity $\exp(\Theta(\log(n)\log(nw/\ep)))$, and this has not been improved except where $w\leq 3$ or $\ep=n^{-\omega(1)}$ or $w=n^{-\omega(1)}$.\\

    However, in other models the picture is not so clear. There has been extensive work on \textit{permutation branching programs}, which are a subset of ordered branching programs that posses additional structure.

    \begin{definition}\label{def:prbp}
        An \textbf{(ordered) permutation branching program} of length $n$, and width $w$ is an ordered branching program where for all $t \in [n]$ and $\sigma \in \zo$, $B_t(\cdot,\sigma)$ is a permutation on $[w]$.
    \end{definition}
    Similarly to how general ordered branching programs can simulate walks on general directed graphs of outdegree 2, permutation branching programs can simulate walks on 2-{\em regular} and {\em consistently labelled} directed graphs $G$. Indeed, for every such graph $G$, the ordered branching program $B^{G,u,v,n}$ defined earlier will be a permutation branching program.

    There are several constructions of deterministic averaging samplers for permutation branching programs that beat the classic Nisan and INW analyses in the constant width regime~\cite{BravermanRaRaYe10,KouckyNiPu11,De11,Steinke12}. We consider the opposite regime, that of \textit{unbounded width permutation branching programs with a single accept state}. This model was introduced by Hoza, Pyne and Vadhan~\cite{HozaPyVa21}, and corresponds to derandomizing walks with \textit{no} constraint on the size of the graph, merely requiring it to be consistently labeled.
    
    In the case of \textit{one-sided} derandomization, they established that deterministic hitters for bounded-width permutation branching programs are also deterministic hitters for unbounded width permutation BPs:
    \begin{proposition}[\cite{HozaPyVa21} Proposition 7.1]\label{hpv:hsg}
        Given $n\in \N$ and $\delta>0$, there is a value $w=O(n^2/\delta)$ such that if $H\subseteq \zo^n$ is a deterministic $\delta$-hitter for permutation branching programs of length $n$ and width $w$, then $H$ is a deterministic $2\delta$-hitter for permutation branching programs of length $n$ and unbounded width (with a single accept state).
    \end{proposition}
    This result, together with an accompanying lower bound (see Claim~\ref{clm:HSGlb}), established that optimal hitters for bounded-width permutation branching programs imply optimal hitters for the unbounded width case, but says nothing about two-sided derandomization. 
    
    In the two-sided regime, they constructed an explicit deterministic averaging $\ep$-sampler for unbounded width permutation branching programs with query complexity $\exp(\tO(\log(n)\log(1/\ep)))$. Moreover, they showed an unconditional lower bound on the query complexity of deterministic averaging samplers of $\exp(\tilde{\Omega}(\log(n)\log(1/\ep)))$. They also showed a random set $Q$ of points in $\zo^n$ fails to produce a deterministic averaging sampler whp unless $|Q|=\exp(\Omega(n))$, so in contrast to the case of general ordered branching programs, they obtained an explicit deterministic sampler with exponentially smaller query complexity than is obtained via the probabilistic method. However, there remained a gap between their upper bound of $\exp(\tO(\log(n)\log(1/\ep)))$ and the lower bound of $(n/\ep)^{\Omega(1)}=\exp(\Omega(\log(n/\ep)))$ on the query complexity of general (possibly non-averaging) deterministic $\ep$-samplers for this model.

    Next, Pyne and Vadhan~\cite{PyneVa21a} constructed an explicit deterministic $\ep$-sampler (in fact a weighted pseudorandom generator, which outputs a fixed linear combination of the queried points) for the model with query complexity $\exp(\tO(\log(n)\sqrt{\log(n/\ep)}+\log(1/\ep)))$. This deterministic sampler obtains smaller query complexity than every deterministic \textit{averaging} sampler when $\ep=n^{-\Omega(1)}$. Thus, they obtained an unconditional separation between the query complexity of deterministic averaging and general deterministic samplers for the model. However, this left a gap between the lower bound on query complexity of $\exp(\Omega(\log(n/\ep)))$ and the upper bound of $\exp(\min\{\tO(\log(n)\sqrt{\log(n/\ep)}+\log(1/\ep)),\tO(\log(n)\log(1/\ep))\})$ required for two-sided derandomization of the model. In addition, their construction was highly involved, and relied on sophisticated results in spectral graph theory~\cite{CKKPPRS18,AhmadinejadKeMuPeSiVa19}, as well as the connection between the INW generator on permutation branching programs and the derandomized square of Rozenman and Vadhan~\cite{RozenmanVa05}.

\subsection{Our Contribution}
Our main result is to resolve the query complexity of derandomizing unbounded-width permutation branching programs.
\begin{theorem}\label{thm:mainNonExplicit}
    There is a non-explicit deterministic nonadaptive $\ep$-sampler for permutation branching programs of length $n$ and unbounded width (with a single accept state) that has query complexity $\poly(n/\ep)$.
\end{theorem}
Thus, we establish the optimal query complexity for deterministic algorithms estimating the fraction of fixed-length walks from $u$ that end at $v$ for arbitrary $u,v$ in an arbitrarily sized consistently-labeled graph. Furthermore, we obtain a deterministic sampler that achieves query complexity $\poly(n)$ for $\ep=1/\poly(n)$, whereas every deterministic \textit{averaging} sampler with these parameters has query complexity $\exp(\Omega(\log^2n))$~\cite{HozaPyVa21}. This gives a simple unconditional separation between averaging samplers and general nonadaptive samplers in the no-randomness regime with respect to a natural computational model.

We prove this result via a reduction from the unbounded-width case to the bounded-width case. We show that an optimal family of samplers for bounded-width permutation branching programs can be used to construct an optimal sampler for unbounded-width ones. Since optimal non-explicit samplers for the bounded-width case exist via the probabilistic method (See Appendix~\ref{app:proofs}), this immediately establishes our result.

We now state the reduction. For the remainder of the paper, rather than working with branching programs with a single accept state $\va\in V_n$, we allow branching programs to have a set $\Va\subseteq V_n$ of accept vertices, where $B(x)=1$ if $B[\vs,x]\in  \Va$. We let $a=|\Va|$ be the number of accept vertices.
\begin{restatable}{theorem}{main}\label{thm:main}
    Let $\fS=\{\samp_{n,w,\ep}\}$ be a family of deterministic $\ep$-samplers $\samp_{n,w,\ep}$ for permutation branching programs of length $n$ and width $w$ such that $\samp_{n,w,\ep}$ has query complexity $\poly(nw/\ep)$. From $\fS$, we can construct a deterministic $\ep$-sampler $\samp'_{n,a,\ep}$ for permutation branching programs of length $n$ and unbounded width with $a$ accept vertices that has query complexity $\poly(na/\ep)$. Moreover, if the samplers $\samp_{n,w,\ep}$ in $\fS$ are explicit then so is $\samp'_{n,a,\ep}$, and if the samplers $\samp_{n,w,\ep}$ in $\fS$ are non-adaptive then so is $\samp'_{n,a,\ep}$.
\end{restatable}
Note that this reduction preserves explicitness and (non-)adaptiveness, so optimal \textit{explicit} deterministic samplers (for instance, optimal explicit PRGs) for the bounded width case imply explicit deterministic samplers for unbounded-width permutation BPs that have optimal space complexity $O(\log(na/\ep))$. Put differently, optimal black-box two-sided derandomization of permutation branching programs with $k$ vertices in \textit{all} layers is no harder than derandomization of permutation branching programs with $k$ accept vertices in the \textit{final} layer and no restriction on width.

We summarize the current known derandomizations for unbounded-width permutation branching programs. An entry of ``Conditional'' means an optimal explicit construction for permutation branching programs of bounded width would imply an explicit construction.
\begin{center}
\begin{tabular}{c|c|c|c}
    Object & Query Complexity & Explicit? & Reference\\\hline
    Det. Hitter & $\exp(O(\log(na/\ep)))$ & Conditional & \cite{HozaPyVa21}\\
    PRG & $\exp(\widetilde{O}(\log(n)\log(a/\ep)))$ & Yes & \cite{HozaPyVa21}\\
    WPRG & $\exp(\tO(\log(n)\sqrt{\log(na/\ep)}+\log(a/\ep)))$ & Yes & \cite{PyneVa21a}\\
    Det. Sampler & $\exp(O(\log(na/\ep)))$ & Conditional & This work.
\end{tabular}
\end{center}

\subsection{Proof Overview}
Our construction is very simple, and in contrast to prior work on the model~\cite{HozaPyVa21,PyneVa21a} the proof uses neither special properties of the INW PRG~\cite{ImpagliazzoNiWi94}, nor results from spectral graph theory.

The key idea behind Theorem~\ref{thm:main} is that compositions of layers of permutation branching programs themselves define permutations. More concretely, fixing a permutation branching program $B$, an input $x\in \zo^{n-i}$ and an accept state $v_f \in \Va$ in the final layer, there is at most \textit{one} state $v$ in layer $i$ such that $B[v,x]=v_f$. With this observation, we can use a sparse set of strings $T\subseteq \zo^n$ to restrict the branching program. For every layer $V_i$, we remove all states $v\in V_i$ where for all $x\in T$, $B[v,x_{1..n-i}]\notin \Va$. Since each element of $T$ can cause at most $a=|\Va|$ vertices in every layer to be included in the restricted program, which we denote $B_T$, we limit the width of $B_T$ to at most $|T|\cdot a$. Furthermore, by adding $n|T|$ dummy states we have that $B_T$ can be computed by a permutation branching program.

We next show that there is a sparse set $T$ such that the restriction induced by $T$ is a good approximation of the original program. We take $T$ to be the set of points queried by a deterministic hitter for permutation branching programs of unbounded width. To obtain $T$ from our hypothesis, we use a result of HPV, which proves that samplers for the bounded-width case are \textit{hitters} for the unbounded-width case (Proposition~\ref{hpv:hsg}). We show that the states not included in the restricted program are unimportant, in that removing all of them simultaneously only changes the acceptance probability $\Pr[B(U_n)=1]$ by at most $\ep/2$.

Then to estimate the acceptance probability of the restricted program, we use a \textit{second} sampler that is good against branching programs of width $|T|\cdot a$, and return the output of the sampler on the restricted program $B_T$. Unfortunately, even if $T$ is explicit it is unclear how to learn $B_T$ given only oracle access to the original program $B$. To avoid having to do so, we construct a way to compute $B_T(x)$ for arbitrary $x\in \zo^n$ given only oracle access to $B$; we apply this procedure whenever the second sampler queries $B_T$. Thus we obtain a good estimate of $\Pr[B_T(U_n)=1]$, which is itself a good estimate of $\Pr[B(U_n)=1]$, so we conclude.

\subsection{Organization}\label{subsec:org}
In Section~\ref{sec:samplertoHSG} we recall that a deterministic $\ep$-sampler for a class of functions implies a deterministic $2\ep$-hitter for that class, and use this to establish the optimal space and query complexity of samplers for unbounded-width permutation BPs. In Section~\ref{sec:HSGtorestrict} we prove that hitters can be used to restrict unbounded-width permutation BPs to bounded width, and that this restriction can be done in a black-box manner, and this restriction is a good approximation of the original program. Then in Section~\ref{sec:main}, we combine these two results and prove the main theorem.

\section{Samplers Imply Hitters}\label{sec:samplertoHSG}
We first recall that an arbitrary deterministic $\ep$-sampler for a model that includes the all-zeroes function (which includes functions computed by permutation branching programs of width at least $2$) induces a deterministic $2\ep$-hitter for the model. We use this to establish tight lower bounds on space and query complexity for deterministic samplers for permutation branching programs of unbounded width. 
\begin{proposition}\label{prop:inducedhsg}
    Let $\F$ be a class of functions $f : \zo^n \ra \zo$ that includes the constant function $f_0(x)=0$. Then if $\A$ is an deterministic $\ep$-sampler for $\F$, the set of queries $Q$ made by $\A$ on the all-zeroes function $f_0$ is a deterministic $2\ep$-hitter for $\F$, and moreover $Q$ is explicit if $\A$ is.
\end{proposition}
\begin{proof}
    Assuming for contradiction this is not the case, there is $f\in \F$ such that $f(Q)=0$ but $\Pr[f(U_n)=1]> 2\ep$. But since the sampler must output a single estimate $\samp^f()=\samp^{f_0}()$ for $f$ and the all-0 program $f_0$ (since the value of both functions on all queried points are identical), it must fail to estimate the acceptance probability of one to within $\ep$, a contradiction.
\end{proof}

We then recall the optimal seed length for deterministic hitters for permutation branching programs of unbounded width.
\begin{restatable}[\cite{HozaPyVa21} Claim 7.3]{claim}{HSGlb}\label{clm:HSGlb}
    Given $n,a\in \N$ and $\ep\in (1/4,0)$ such that $1/2> \ep/a\geq 2^{-n}$, let $H\subseteq \zo^n$ be a deterministic $\ep$-hitter for permutation branching programs of unbounded width and at most $a$ accept vertices. Then $|H|=(na/\ep)^{\Omega(1)}=\exp(\Omega(\log(na/\ep)))$.
\end{restatable}
We recall the proof in Appendix~\ref{app:proofs}.
From this, we can derive the optimal space and query complexity of a sampler.
\begin{corollary}\label{cor:samplerLB}
    Given $1/8> \ep\geq 2^{-n}$ and $n,a\in\N$, let $\A$ be a deterministic $\ep$-sampler for permutation branching programs of length $n$ and unbounded width with at most $a$ accept vertices, where $\ep/a\geq 2^{-n}$. Then $\A$ has query complexity $(na/\ep)^{\Omega(1)}$. Moreover, if $\A$ is explicit it has space complexity $s=\Omega(\log(na/\ep))$.
\end{corollary}
\begin{proof}
    We apply Proposition~\ref{prop:inducedhsg} to $\A$ and obtain a deterministic $2\ep$-hitter $Q\subseteq \zo^n$ for permutation branching programs of length $n$ with at most $a$ accept vertices. By Claim~\ref{clm:HSGlb} we obtain $|Q|=(na/\ep)^{\Omega(1)}$ which establishes the claimed bound on query complexity. Furthermore if $\A$ is explicit and has space complexity $s$, it must run in time $2^{O(s)}$ since it is required to halt, and thus its query complexity $|Q|$ is at most $2^{O(s)}$. Combined with the lower bound on $|Q|$, we have that $s=\Omega(na/\ep)$.
\end{proof}

\section{Hitters Induce Bounded-Width Approximators}\label{sec:HSGtorestrict}
We next show that, given a permutation branching program $B$ and a sufficiently good deterministic hitter, the set of states $v\in V_i$ for which there is a hitter output whose prefix reaches an accept state starting from $v$ forms an approximator of the original program. To show this, we define the program ``cut out'' by a deterministic hitter. Then we show that we can evaluate this approximator program on any input given only oracle access to $B$. 
\begin{definition}
    Given a set $H\subseteq \zo^n$ and a permutation branching program $B$ of length $n$ with $a$ accept vertices $\Va$ and vertices $V_0,\dots,V_n$, let the hit states in layer $i$ be 
    \[K_i =\{v\in V_i:\exists x\in H \text{ s.t. } B[v,x_{1..n-i}]\in V_a\}.\]
    WLOG pad all such sets to have size $K=\max_{i=0}^n |K_i|$, where all transitions from padding states in $K_i$ do not lead to $K_{i+1}$.
    The \textbf{induced hit program} $B_H$ is the length $n$ permutation branching program with states in layer $i$ given by $\{K_i\} \cup \{\{0,\ldots,n\}\times [K]\}$, where we identify states in $\{0,\ldots,n\}\times [K]$ by $(j,v)$. For $v\in K_i$, define the transition function
    \[(B_H)_i(v,b) = \begin{cases}
    B_i(v,b) & B_i(v,b) \in K_{i+1}\\
    (i,v) & \text{otherwise.}\end{cases}\]
    Then greedily define transitions for $\{(i,v):v\in [K]\}$ to maintain the permutation property. For all states $(j,v)$ for $j\neq i$, let $(B_H)_i((j,v),b)=(j,v)$.
\end{definition}
We next show that the width of the induced hit program is bounded by the size of the domain of $H$ (and thus its seed length). This will allow us to derandomize the induced hit program as a standard bounded-width permutation branching program.
\begin{lemma}\label{lem:widthbound}
    Given $H\subseteq \zo^n$ and a permutation branching program $B$ of length $n$ with $a$ accept vertices, the width of the induced hit program $B_H$ is at most $|H|\cdot (n+2)\cdot a$.
\end{lemma}
To prove this, we require a proposition essentially showing that composing multiple layers of a permutation branching program produces a permutation branching program of higher degree. This is the only element of the proof that uses the fact that $B$ is a permutation, rather than regular, branching program.
\begin{proposition}\label{prop:pbpdist}
    For every permutation branching program $B$, for every distinct $v,v'\in V_i$ and $\sigma \in \zo^k$ so that $i+k\leq n$, $B[v,\sigma]\neq B[v',\sigma]$.
\end{proposition}
\begin{proof}
    We prove this by induction on $k$. The base case of $k=0$ is vacuously true. Assuming it holds for $k$, let $B$ be an arbitrary permutation branching program and $v,v'\in V_i$ arbitrary distinct states. Let $\sigma\in \zo^{k+1}$ be arbitrary. From the permutation property it must be the case that $u_1=B[v,\sigma_1]\neq B[v',\sigma_1]=u_2$, so $B[v,\sigma]=B[u_1,\sigma_{2..k}]\neq B[u_2,\sigma_{2..k}]=B[v',\sigma]$
    where the inequality follows from the inductive step, and since $\sigma$, $B$ and $v,v'$ were arbitrary we conclude.
\end{proof}
We can then prove Lemma~\ref{lem:widthbound}.
\begin{proof}
    It suffices to show that the number of included states of the original program satisfies $|K_i|\leq |H|\cdot a$ for all $i\in \{0,\ldots,n-1\}$, since the width of $B_H$ is bounded by $(n+2)\cdot K=(n+2)\cdot \max_{i=0}^n|K_i|$. For every fixed accept state $u \in \Va$, there are at most $|H|$ states $v\in V_i$ such that there exists $x\in H$ such that $B[v,x_{1..n-i}] = u$ by Proposition~\ref{prop:pbpdist}, so we conclude via a union bound.
\end{proof}

The induced hit program is well defined for every $H\subseteq \zo^n$. However, we wish to show that the program induced by a sufficiently good deterministic hitter is a close approximation of the original permutation branching program.
\begin{lemma}\label{lem:inducedLA}
    Let $H\subseteq \zo^n$ be a deterministic $\delta/na$-hitter for permutation branching programs of length $n$ and unbounded width with a single accept state. Then for every permutation branching program $B$ of length $n$ and unbounded width with at most $a$ accept vertices, the induced hit program $B_H$ satisfies
    \[|\Pr[B_H(U_n)=1]-\Pr[B(U_n)=1]|\leq \delta.\]
\end{lemma}
\begin{proof}
    Let $\Vn$ be the set of states of $B$ not included in the induced hit program $B_H$. For every $v\in \Vn$ in layer $n-k$, using the fact that $H$ is a deterministic $\delta/na$-hitter for branching programs of length $n$, and hence for length $n-k\leq n$ since branching programs can ignore bits, we obtain
    \[\Pr[B[v,U_{k}]\in \Va]= \sum_{u\in \Va}\Pr[B[v,U_{k}]=u] \leq a\cdot \frac{\delta}{na}.\] 
    If $\vs\in \Vn$ then $B_H$ is the all zeroes program and the above implies $\Pr[B(U_n)=1]\leq \delta$ so we are done. Thus assume that $\vs\notin \Vn$.
    For arbitrary $x\in \zo^n$ such that $B(x)\neq B_H(x)$, it must be the case that $B(x)=1$ while $B_H(x)=0$, i.e. $B$ passes through some element of $\Vn$ in its computation on $x$ and $B[\vs,x]\in \Va$. Therefore,
    \begin{align*}
        |\Pr[B_H(U_n)=1]-\Pr[B(U_n)=1]| &\leq \Pr_{x\la U_n}[B(x)\neq B_H(x)]\\
        &=\Pr_{x\la U_n}\left[ \left(B[\vs,x]\in \Va \right)\bigwedge \left(\bigvee_{i=1}^n B[\vs,x_{1..i}]\in \Vn\right)\right]\\ 
        &\leq \sum_{i=1}^n \sum_{v\in V_i\cap \Vn}\Pr_{x\la U_{n}}\left[(B[\vs,x_{1..i}]=v) \wedge (B[v,x_{i+1..n}]\in \Va)\right]\\
        &= \sum_{i=1}^n \sum_{v\in \Vn\cap V_i}\Pr[B[\vs,U_i]=v]\cdot\Pr[ B[v,U_{n-i}]\in \Va]\\
        &\leq \sum_{i=1}^n \left(\sum_{v\in \Vn\cap V_i}\Pr[B[\vs,U_i]=v]\right)\cdot\frac{\delta}{n}\\
        &\leq \sum_{i=1}^n 1\cdot \frac{\delta}{n} = \delta.
    \end{align*}
\end{proof}

Finally, we show that given $H$, we can evaluate an arbitrary input on the induced hit program.
\begin{lemma}\label{lem:hiteval}
    Given a permutation branching program $B$ of length $n$ and a set $H\subseteq \zo^n$, for every $x\in \zo^n$ we have
    \[B_H(x) = \bigwedge_{i=0}^n\left(\bigvee_{y\in H}B(x_{1..i}||y_{i+1..n})\right)\]
    where $||$ denotes string concatenation and $x_{1..0}$ and $y_{n+1..n}$ are the empty string.
\end{lemma}
\begin{proof}
    Fix arbitrary $x\in \zo^n$ and let $v_i=B[\vs,x_{1..i}]$ for all $i\in\{0,\ldots,n\}$.
    First suppose the RHS evaluates to $1$. For $i\in\{0,\ldots,n\}$ we have $1=\bigvee_{y\in H}B(x_{1..i}||y_{i+1..n})$, so there is some $y\in H$ such that $B[v_i,y_{i+1..n}]\in \Va$, which is precisely the condition for including $v_i$ in the induced hit program $B_H$, and this holds for every $i$, so $B_H(x)=1$. 
    Now suppose the RHS evaluates to $0$. Fixing the least $i$ such that $\bigvee_{y\in H}B(x_{1..i}||y_{i+1..n})=0$, we have that $v_i$ is not included in $B_H$ and so $B[v_{i-1},x_i]=u\notin K_i$. Since $u$ is always subsequently wired to itself and marked as reject in the final layer we have $B_H(x)=0$.
\end{proof}
Note that this implies that we can evaluate $B_H(x)$ given $x$ and oracle access to $B$, and this procedure is explicit if $H$ is.

\section{Putting it All Together}\label{sec:main}
We can now go from samplers for bounded-width permutation branching programs to samplers for unbounded-width permutation branching programs. We follow the outline in the proof sketch in Section~\ref{subsec:org}. First, we use an optimal sampler for the bounded-width case to generate an optimal deterministic hitter for the bounded-width case, which implies an optimal deterministic hitter $H$ for the unbounded-width case. Then we use a second sampler and evaluate it on the induced hit program $B_H$, and by choosing the parameters for the second sampler appropriately obtain an accurate estimate of $\Pr[B_H(U_n)=1]$ and thus $\Pr[B(U_n)=1]$.
\main*
\begin{proof}
    By assumption, we have a deterministic $\ep$-hitter for permutation branching programs of length $n$ and width $w=O(n^3a/\ep)$, where this $w$ is that obtained from Proposition~\ref{hpv:hsg} with $n=n$ and $\delta=\ep/(8na)$. Applying Proposition~\ref{prop:inducedhsg}, we obtain a deterministic $\ep/(4na)$-hitter $H\subseteq \zo^n$ for permutation BPs of length $n$ and width $w=O(n^3a/\ep)$, with $|H|=(na/\ep)^{O(1)}=\exp(O(\log(na/\ep)))$. Applying Proposition~\ref{hpv:hsg}, we have that $H$ is a deterministic $\ep/2na$-hitter for permutation branching programs of unbounded width with a single accept state.
    
    Now let $\A=\samp_{n,w',\ep/2}$ be an $\ep/2$-sampler for permutation branching programs of length $n$ and width $w'=|H|\cdot (n+2)\cdot a=\poly(na/\ep)$. By assumption, $\A$ has query complexity $\poly(na/\ep)$.
    
    Finally, given an arbitrary permutation branching program $B$ with at most $a$ accept vertices, define $\samp^{'B}_{n,a,\ep}()=\A^{B_H}$(), where whenever $\A$ queries the value of $B_H$ on $x\in \zo^n$, we apply Lemma~\ref{lem:hiteval}, so $\samp^{'B}_{n,a,\ep}$ queries
    \[\{(x_{1..i}||y_{1..n-i}):i\in \{0,\ldots,n\}, y\in H\}.\]
    Thus $\samp'_{n,a,\ep}$ has query complexity that equals the query complexity of $\A$ times $(n+1)\cdot |H|$, for a total query complexity of $\poly(na/\ep)$. If $\fS=\{\samp_{n,w,\ep}\}$ is non-adaptive (i.e. the queries made by $\A$ do not depend on $B_H$ and thus $B$) then $\samp'_{n,a,\ep}$ is. Finally, if $\fS=\{\samp_{n,w,\ep}\}$ is explicit then $\samp'_{n,a,\ep}$ is by definition.\\
    
    Finally, we prove $\samp_{n,a,\ep}'$ is an $\ep$-sampler. We have that $B_H$ is a permutation branching program of length $n$ and width at most $|H|\cdot (n+2)\cdot a$ by Lemma~\ref{lem:widthbound}, so by our choice of parameters $\A$ is an $\ep/2$-sampler for $B_H$, i.e. 
    $|\A^{B_H}() - \Pr[B_H(U_n)=1]| \leq \ep/2$.
    Then we conclude by the triangle inequality:
    \begin{align*}
        \MoveEqLeft{|\samp^{'B}_{n,a,\ep}() - \Pr[B(U_n)=1]|}\\
        &= |\A^{B_H}()-\Pr[B(U_n)=1]| && \text{(Lemma~\ref{lem:hiteval})}\\
        &\leq |\A^{B_H}()-\Pr[B_H(U_n)=1]| +|\Pr[B_H(U_n)=1] - \Pr[B(U_n)=1]|\\
        &\leq |\A^{B_H}()-\Pr[B_H(U_n)=1]| + \ep/2 && \text{(Lemma~\ref{lem:inducedLA})}\\
        &\leq \ep
    \end{align*}
    and since $B$ was arbitrary we obtain the result.
\end{proof}

\bibliographystyle{alpha}
\bibliography{references}
\appendix

\section{The Optimal Seed Length For Deterministic Hitters}\label{app:proofs}
We prove the optimal seed length for deterministic hitters for unbounded-width permutation branching programs. The proof is identical to that given in Hoza et al.~\cite{HozaPyVa21} except we explicitly consider the number of accept vertices in the final layer.
\HSGlb*
\begin{proof}
We prove the bounds on $n$ and $\ep/a$ separately:
\begin{enumerate}
    \item If $|H| \leq n-1$, there is some nonzero vector $z \in \mathbb{F}_2^n$ such that for every $x$,
    $$
        \bigoplus_{i = 1}^n z_i \cdot H(x)_i = 0.
    $$
    The function $B(x) = \mathbb{I}[\oplus_{i = 1}^n z_i \cdot x_i = 1]$ can be computed by a permutation branching program (of width $2$) with a single accept state, and $\Pr[B(U_n) = 1] \geq 1/2$, but $B(x)=0$ for all $x\in H$, a contradiction.
    
    \item If $|H| \leq a/4\ep$, there are at most $a/4\ep$ distinct length $l=\lceil \log(\ep/a)\rceil-1$ prefixes of elements of $H$. Letting $B$ be a permutation branching program that reaches the $i$th state in the final layer on input $x$ if
    \[\sum_{j=1}^{l}2^{j-1}\cdot x_j=i,\]
    we have that the final state of $B(U_n)$ is distributed uniformly over $\{0,\ldots,2^l-1\}$. Then choosing $a$ distinct states in the final layer that are not reached by every output of $H$ (where we use that $2^l-a/4\ep\geq a$) and marking them as accept, we obtain that $\Pr[B(U_n)=1] \geq a(2\ep/a)=2\ep$ whereas $B(x)=0$ for all $x\in H$, a contradiction.
\end{enumerate}
\end{proof}

We also note the optimal seed length for hitters for general ordered branching programs (which include permutation branching programs of bounded width). Together with the result converting deterministic samplers into hitters (Proposition~\ref{prop:inducedhsg}), this establishes the query complexity of two-sided derandomization for the model is $(nw/\ep)^{\Omega(1)}$.
\begin{claim}
Given $n,w\in \N$ and $1/2>\ep\geq 2^{-n}$, let $H\subseteq \zo^n$ be a deterministic $\ep$-hitter for ordered branching programs of width $w$ and length $n$. Then $|H|=(nw/\ep)^{\Omega(1)}$.
\end{claim}
\begin{proof}
We prove the bounds on $w$ and $1/\ep$ separately. 
The bound $|H|=n^{\Omega(1)}$ is identical to the same bound in the prior proof.
\begin{enumerate}
    \item If $|H|\leq 1/2\ep$, there is some string $\sigma$ of length $l=\lceil \log(1/\ep)\rceil-1$ such that $x_{1..l}\neq \sigma$ for every $x\in H$. Then let $B$ be a width-$2$ ordered branching program where
    \[B(x)=1 \iff x_{1..l}=\sigma\]
    (note that we use the program is non-regular to keep the width at $2$), we have $\Pr[B(U_n)=1]\geq 2\ep$ but $B(x)=0$ for all $x\in H$, a contradiction.
    
    \item If $|H| \leq w/2$, there are at most $w/2$ distinct length $l=\lceil \log(w)\rceil$ prefixes of elements of $H$. Letting $B$ be a branching program that reaches the $i$th state in the final layer on input $x$ if
    \[\sum_{j=1}^{l}2^{j-1}\cdot x_j=i,\]
    we have that the final state of $B(U_n)$ is distributed uniformly over $\{0,\ldots,2^l-1\}$. Then choosing $w/2$ distinct states in the final layer that are not reached by every $x\in H$ and marking them as accept states, we obtain that $\Pr[B(U_n)=1] \geq 1/2$ whereas $B(x)=0$ for all $x\in H$, a contradiction.
\end{enumerate}
\end{proof}

We recall the existence of non-explicit deterministic averaging samplers for ordered branching programs of bounded width. Since ordered branching programs are a superset of permutation branching programs, this likewise implies the existence of optimal non-explicit averaging samplers for bounded-width permutation branching programs. Together with Theorem~\ref{thm:main} this implies Theorem~\ref{thm:mainNonExplicit}.
\begin{proposition}
    There is a family $\fS=\{\samp_{n,w,\ep}\}$ of (non-explicit) deterministic $\ep$-samplers $\samp_{n,w,\ep}$ for ordered branching programs of length $n$ and width $w$ such that $\samp_{n,w,\ep}$ has query complexity $\poly(nw/\ep)$.
\end{proposition}
\begin{proof}
First note that a branching program of length $n$ and width $w$ has a description using $k=\poly(nw)$ bits, so there are at most $2^k$ such programs. Now fix $n,w\in \N$ and $\ep>0$ and consider a random set $Q$ of size $|Q|=3k/\ep$ where $Q_i$ is a random independently chosen element of $\zo^n$ for all $i$. Fixing an arbitrary ordered branching program $B$ of length $n$ and width $w$, let $\mu=\Pr[B(U_n)=1]$ be its accept probability and WLOG assume $\mu\geq 1/2$, since an additive estimate of the accept probability implies an equivalent estimate of the reject probability. Then let $Y_i=B(Q_i)$ be the random variable that is $1$ when $B$ accepts on $Q_i$. We have $Y_i\in [0,1]$ and they are independent for all $i$. Applying a Chernoff bound over the randomness of the strings in $Q$, we obtain for all $\delta\in(0,1)$
\[\Pr\left[\left|\frac{1}{|Q|}\sum_i Y_i - \mu\right| \geq \delta\mu\right] \leq 2\exp(-|Q|^2\delta^2\mu/3).\]
Then choosing $\delta=\ep$ we obtain
\[\Pr\left[\left|\frac{1}{|Q|}\sum_i Y_i - \mu\right| \geq \ep\right] \leq 2\exp(-3k^2).\]
By a union bound, the probability that a random set $Q$ of the chosen size fails to be a deterministic averaging sampler for at least one of the $2^k$ length $n$, width $w$ branching programs is at most $2\exp(k-3k^2)<1$. Thus there exists some $Q_{good}$ that is good for all such programs, so $Q_{good}$ generates a deterministic averaging $\ep$-sampler $\samp_{n,w,\ep}$ for ordered branching programs of length $n$ and width $w$. By construction, $|Q_{good}|=\poly(nw/\ep)$.
\end{proof}

\end{document}